\documentclass{eptcs}
\providecommand{\event}{MSFP 2018} 
\usepackage{breakurl}

\makeatletter
\@ifundefined{lhs2tex.lhs2tex.sty.read}%
  {\@namedef{lhs2tex.lhs2tex.sty.read}{}%
   \newcommand\SkipToFmtEnd{}%
   \newcommand\EndFmtInput{}%
   \long\def\SkipToFmtEnd#1\EndFmtInput{}%
  }\SkipToFmtEnd

\newcommand\ReadOnlyOnce[1]{\@ifundefined{#1}{\@namedef{#1}{}}\SkipToFmtEnd}
\usepackage{amstext}
\usepackage{amssymb}
\usepackage{stmaryrd}
\DeclareFontFamily{OT1}{cmtex}{}
\DeclareFontShape{OT1}{cmtex}{m}{n}
  {<5><6><7><8>cmtex8
   <9>cmtex9
   <10><10.95><12><14.4><17.28><20.74><24.88>cmtex10}{}
\DeclareFontShape{OT1}{cmtex}{m}{it}
  {<-> ssub * cmtt/m/it}{}

\DeclareFontShape{OT1}{cmtt}{bx}{n}
  {<5><6><7><8>cmtt8
   <9>cmbtt9
   <10><10.95><12><14.4><17.28><20.74><24.88>cmbtt10}{}
\DeclareFontShape{OT1}{cmtex}{bx}{n}
  {<-> ssub * cmtt/bx/n}{}

\newcommand{\Conid}[1]{\mathit{#1}}
\newcommand{\Varid}[1]{\mathit{#1}}
\newcommand{\anonymous}{\kern0.06em \vbox{\hrule\@width.5em}}

\newcommand{\bind}{\mathbin{>\!\!\!>\mkern-6.7mu=}}

\usepackage{polytable}

\@ifundefined{mathindent}%
  {\newdimen\mathindent\mathindent\leftmargini}%
  {}%

\def\resethooks{%
  \global\let\SaveRestoreHook\empty
  \global\let\ColumnHook\empty}
\newcommand*{\savecolumns}[1][default]%
  {\g@addto@macro\SaveRestoreHook{\savecolumns[#1]}}
\newcommand*{\restorecolumns}[1][default]%
  {\g@addto@macro\SaveRestoreHook{\restorecolumns[#1]}}
\newcommand*{\aligncolumn}[2]%
  {\g@addto@macro\ColumnHook{\column{#1}{#2}}}

\resethooks

\newcommand{\onelinecommentchars}{\quad-{}- }
\newcommand{\commentbeginchars}{\enskip\{-}
\newcommand{\commentendchars}{-\}\enskip}

\newcommand{\visiblecomments}{%
  \let\onelinecomment=\onelinecommentchars
  \let\commentbegin=\commentbeginchars
  \let\commentend=\commentendchars}

\newcommand{\invisiblecomments}{%
  \let\onelinecomment=\empty
  \let\commentbegin=\empty
  \let\commentend=\empty}

\visiblecomments

\newlength{\blanklineskip}
\newcommand{\hsindent}[1]{\quad}
\let\hspre\empty
\let\hspost\empty

\EndFmtInput
\makeatother
\ReadOnlyOnce{polycode.fmt}%
\makeatletter

\newcommand{\hsnewpar}[1]%
  {{\parskip=0pt\parindent=0pt\par\vskip #1\noindent}}

\newcommand{\hscodestyle}{}

\newcommand{\sethscode}[1]%
  {\expandafter\let\expandafter\hscode\csname #1\endcsname
   \expandafter\let\expandafter\endhscode\csname end#1\endcsname}

  {\par\noindent
   \advance\leftskip\mathindent
   \hscodestyle
   \let\\=\@normalcr
   \let\hspre\(\let\hspost\)%
   \pboxed}%
  {\endpboxed\)%
   \par\noindent
   \ignorespacesafterend}

  {\hsnewpar\abovedisplayskip
   \advance\leftskip\mathindent
   \hscodestyle
   \let\hspre\(\let\hspost\)%
   \pboxed}%
  {\endpboxed%
   \hsnewpar\belowdisplayskip
   \ignorespacesafterend}

  {\hsnewpar\abovedisplayskip
   \advance\leftskip\mathindent
   \hscodestyle
   \let\\=\@normalcr
   \(\pboxed}%
  {\endpboxed\)%
   \hsnewpar\belowdisplayskip
   \ignorespacesafterend}

\newcommand{\plainhs}{\sethscode{plainhscode}}

\plainhs

  {\hsnewpar\abovedisplayskip
   \advance\leftskip\mathindent
   \hscodestyle
   \let\\=\@normalcr
   \(\parray}%
  {\endparray\)%
   \hsnewpar\belowdisplayskip
   \ignorespacesafterend}

  {\parray}{\endparray}

  {\(\parray}{\endparray\)}

\def\codeframewidth{\arrayrulewidth}
\RequirePackage{calc}

  {\parskip=\abovedisplayskip\par\noindent
   \hscodestyle
   \arrayrulewidth=\codeframewidth
   \tabular{@{}|p{\linewidth-2\arraycolsep-2\arrayrulewidth-2pt}|@{}}%
   \hline\framedhslinecorrect\\{-1.5ex}%
   \let\endoflinesave=\\
   \let\\=\@normalcr
   \(\pboxed}%
  {\endpboxed\)%
   \framedhslinecorrect\endoflinesave{.5ex}\hline
   \endtabular
   \parskip=\belowdisplayskip\par\noindent
   \ignorespacesafterend}

\newcommand{\framedhslinecorrect}[2]%
  {#1[#2]}

  {\(\def\column##1##2{}%
   \let\>\undefined\let\<\undefined\let\\\undefined
   \newcommand\>[1][]{}\newcommand\<[1][]{}\newcommand\\[1][]{}%
   \def\fromto##1##2##3{##3}%
   }{\) }%

  {\let\orighscode=\hscode
   \let\origendhscode=\endhscode
   \def\endhscode{\def\hscode{\endgroup\def\@currenvir{hscode}\\}\begingroup}
   \orighscode\def\hscode{\endgroup\def\@currenvir{hscode}}}%
  {\origendhscode
   \global\let\hscode=\orighscode
   \global\let\endhscode=\origendhscode}%

\makeatother
\EndFmtInput
\ReadOnlyOnce{forall.fmt}%
\makeatletter

\let\HaskellResetHook\empty
\newcommand*{\AtHaskellReset}[1]{%
  \g@addto@macro\HaskellResetHook{#1}}
\newcommand*{\HaskellReset}{\HaskellResetHook}

\newcommand\hsforall{\global\let\hsdot=\hsperiodonce}
\newcommand*\hsperiodonce[2]{#2\global\let\hsdot=\hscompose}
\newcommand*\hscompose[2]{#1}

\AtHaskellReset{\global\let\hsdot=\hscompose}

\HaskellReset

\makeatother
\EndFmtInput

\usepackage[all]{xy}
\usepackage{amsmath}
\usepackage{amsthm}
\usepackage{cleveref}
\usepackage{graphicx}

\newcommand{\Set}{\mathbb{S}}
\newcommand{\Fin}{\mathbb{F}}
\newcommand{\op}{{\mathrm{op}\!}}
\newcommand{\str}{\mathrm{str}}
\newcommand{\id}{\mathrm{id}}
\newcommand{\Id}{\mathrm{Id}}
\newcommand{\Func}{\left[\Fin,\Set\right]}
\newcommand{\Prof}{\left[\Fin^\op \times \Fin, \Set\right]}
\newcommand{\ProfS}{{\Prof_\mathrm{s}}}
\newcommand{\Hom}{\mathrm{Hom}}

\newcommand{\lstar}[1]{{{#1}_{*}}}
\newcommand{\ustar}[1]{{{#1}^{*}\!}}
\newcommand{\ladm}[1]{{#1}_{!}}
\newcommand{\blank}{\raisebox{.4pt}{-}}
\newcommand{\inc}{i}
\newcommand{\ev}{\mathrm{ev}}
\newcommand{\curry}[1]{{\left\lfloor{#1}\right\rfloor}}
\newcommand{\Mon}[1]{\mathrm{Mon}\left({#1}\right)}
\newcommand{\inv}[1]{{\left({#1}\right)^{-1}}}

\newtheorem{definition}{Definition}
\newtheorem{lemma}{Lemma}
\newtheorem{proposition}{Proposition}
\newtheorem{theorem}{Theorem}

\newtheorem{remark}{Remark}
\newtheorem*{ack}{Acknowledgments}

\title{Relating Idioms, Arrows and Monads from Monoidal Adjunctions}
\author{Exequiel Rivas
\institute{$\pi.r^2$ team, INRIA - IRIF\\ Paris, France}
\email{er@irif.fr}
}
\def\titlerunning{Relating Idioms, Arrows and Monads}
\def\authorrunning{E. Rivas}

\begin{document}
\maketitle
\begin{abstract}
  We revisit once again the connection between three notions of
  computation: \emph{monads}, \emph{arrows} and \emph{idioms} (also
  called \emph{applicative functors}). We employ monoidal categories
  of finitary functors and profunctors on finite sets as models of
  these notions of computation, and develop the connections between
  them through adjunctions. As a result, we obtain a categorical
  version of Lindley, Yallop and Wadler's characterisation of monads
  and idioms as arrows satisfying an isomorphism.
\end{abstract}

\section{Introduction}\label{sec:introduction}
The semantic study of computational effects has not only provided
tools for reasoning about effectful programs, but also introduced a
methodology for organising functional code around fundamental
interfaces coming from mathematics. This is the story of monads, which
in the beginning were introduced by Moggi to model effectful
computations, but soon after that Wadler internalised them in a
functional programming language to structure code. In this way, monads
became an interface capturing a pattern, which was later expressed as
a \emph{type-class}, and chosen to be the interface for the basic IO
mechanism in Haskell. Over time, new interfaces were defined, each one
providing different levels of control on how to combine computations.
We are interested in two of these interfaces that emerged:
\emph{arrows} and \emph{idioms}. Arrows were introduced by
Hughes~\cite{Hug:GMA} and idioms, or \emph{applicative functors}, by
McBride and Paterson~\cite{MBP:APE}.  Even when arrows and idioms
differ from monads, they still share some basic form. In previous
work~\cite{RJ:NCM}, we gave an unified presentation of these three
interfaces in terms of monoids in monoidal categories. In this paper
we take a first step in addressing the connection between these
interfaces by looking at their monoidal structures and relating them
by adjunctions.

Instead of starting from scratch, we build on previous results
relating these structures. Ten years ago, in MSFP 2008, Lindley,
Yallop and Wadler presented an article connecting idioms, arrows and
monads. In a nutshell, we can summarise their result in the following
diagram of embeddings
\[
\xymatrix@C+=1cm@M+=0.5cm{
\textrm{Idioms} \ar@{^{(}->}[r] & \textrm{Arrows} \ar@{^{(}->}[r] & \textrm{Monads}
}
\]
together with the equations
\begin{align}
\textrm{Idiom} &= \textrm{Arrow} + (A\leadsto B \cong 1\leadsto (A \to B)),\label{eq:isoidiomintro}\\
\textrm{Monad} &= \textrm{Arrow} + (A\leadsto B \cong A\to (1 \leadsto B)).\label{eq:isomonadintro}
\end{align}
These equations explain how to see idioms and monads as arrows in
which a particular isomorphism holds. The method they used to
establish these formulas was purely syntactical. Basically, they
proposed theories for describing different versions of
$\lambda$-calculus with effects, and used a notion of theory
translation to show the relations.

This article is a first step towards understanding these equalities
from a semantic point of view. In what follows, we deconstruct these
characterisations by using the underlying structure of idioms, monads
and arrows: functors and strong profunctors.
\begin{remark}
  In order to avoid size issues, we restrict ourselves to modelling
  the notions of computation by finitary functors, and their
  profunctor counterpart, which are strong endoprofunctors on finite
  sets. It should be possible to replace finite sets by a general
  category $\mathbb{C}$ together with an inclusion functor into sets,
  subject to certain conditions, such as density.
\end{remark}
Denoting the category of finite sets by $\Fin$ and the category of
sets by $\Set$, we aim to express
\Cref{eq:isoidiomintro,eq:isomonadintro} in terms of functor
categories
\[
\Func\qquad\mbox{and}\qquad \ProfS
\]
where ${\blank}^\op$ denotes the opposite category and the subscript
``$\mathrm{s}$'' represents the restriction to strong profunctors. The idea that we
develop in the article follows from the next observation. Given a
strong endoprofunctor $P$, we can fix its first parameter and obtain a
finitary functor:
\[
\ustar{P}(X) = P(1, X),\quad\mbox{where $1$ denotes the terminal object of $\Fin$}.
\]
Conversely, there are two ways in which a finitary functor $F : \Fin
\to \Set$ might be presented as a strong endoprofunctor on $\Fin$:
\[
\ladm{F}(X,Y) = F(X\to Y),\qquad \lstar{F}(X,Y) = X \to F Y.
\]
Taking $P(X,Y) = X\leadsto Y$ (the underlying profunctor of an
arrow), we can express the isomorphisms
\[
A\leadsto B \cong 1\leadsto (A \to B),\qquad A\leadsto B \cong A\to (1 \leadsto B)
\]
simply as the equivalences
\[
P(A,B) \cong P(1, A\to B),\qquad P(A,B) \cong A \to P(1, B)
\]
or directly as profunctor equivalences in terms of $\ustar{\blank}$,
$\ladm{\blank}$ and $\lstar{\blank}$
\[
P \cong \ladm{(\ustar{P})},\qquad P \cong \lstar{(\ustar{P})}.
\]%
In the following sections, a pair of adjunctions between
$\ladm{\blank}$, $\ustar{\blank}$ and $\lstar{\blank}$ are introduced,
and as a result we obtain that the map $P \mapsto
\ladm{\left(\ustar{P}\right)}$ extends to a comonad and the map $P
\mapsto \lstar{\left(\ustar{P}\right)}$ extends to a monad. The
mentioned adjunctions are monoidal, and thus establish equivalences
between the categories of monoids. In this way, we get a first
approximation to Lindley et al.'s result from a categorical point of
view.

This article might be seen as a continuation of previous articles in
the spirit of \emph{notions of computation as monoids}. The starting
point was an article of the author with
Jaskelioff~\cite{RJ:NCM}. Later, the adjunctions between
$\ladm{\blank}$, $\ustar{\blank}$ and $\lstar{\blank}$ were
exploited in the context of non-monadic effect
handlers~\cite{PSR:HNMC}, although their relation to the
\Cref{eq:isoidiomintro,eq:isomonadintro} by using the derived
monad/comonad and their corresponding monoidal structure was not
treated.

\newtheorem*{contribution}{Contributions}
\begin{contribution}
  The main contributions of this paper are:
  \begin{itemize}
    \item We give a detailed and more mathematical introduction to the
      adjunctions presented in previous work~\cite{PSR:HNMC}.
    \item We recognise in these adjunctions the essence of Lindley,
      Yallop and Wadler's result from a semantic point of view.
  \end{itemize}
\end{contribution}

The rest of the article is organised as follows. In
Section~\ref{sec:background}, we review some background material
needed to express the adjunctions and relate monoidal structures to
arrows, idioms and monads. Then, in Section~\ref{sec:adjunctions}, we
present adjunctions between the categories representing effects, and
introduce some properties they have. In Section~\ref{sec:equivalences}
we establish \Cref{eq:isoidiomintro,eq:isomonadintro} from a
categorical argument, explaining Lindley et al.'s result from a
semantic point of view. Finally, in Section~\ref{sec:conclusions}, we
conclude and discuss further work.

\section{Background}\label{sec:background}

We assume the reader is familiar with basic concepts of category
theory, including adjunctions. We briefly introduce some categorical
topics we need, such as monoidal categories, partly to be
self-contained and partly to fix notations. An exception to this is
the use of coends, which are not explained here mainly for brevity
reasons. In case the reader is not familiar with them, coends might be
(informally) seen as existential or $\Sigma$ types. The usual
reference for all these topics is Mac Lane's book~\cite{MacL:CWM}, and
a complete reference for (co)ends is Loregian's
notes~\cite{Lor:TCC}. An introduction to these topics oriented for the
functional programming community might be found in~\cite{RJ:NCM}. We
try to follow standard categorical notation for most of the text. We
denote exponentials by $A \to B$, the evaluation morphism by $\ev : (A
\to B) \times A \to B$ and the currying of a morphism $f : A \times B
\to C$ as $\curry{f} : A \to (B \to C)$. Note that exponentials of
$\Fin$ coincide with exponentials of $\Set$ and their respective
hom-sets, and so we treat them ambiguously.

\subsection{Implementation of notions of computation}
Before starting with the categorical preliminaries, we remind the
reader of the basic interfaces we are treating, as in the rest of the
text we will not put emphasis on the programming side.

The first interface we are interested in are monads, which are abstracted
by the following type-class.
\begin{hscode}\SaveRestoreHook
\column{B}{@{}>{\hspre}l<{\hspost}@{}}%
\column{3}{@{}>{\hspre}l<{\hspost}@{}}%
\column{5}{@{}>{\hspre}l<{\hspost}@{}}%
\column{13}{@{}>{\hspre}l<{\hspost}@{}}%
\column{E}{@{}>{\hspre}l<{\hspost}@{}}%
\>[3]{}\mathbf{class}\;\Conid{Functor}\;\Varid{f}\Rightarrow \Conid{Monad}\;\Varid{f}\;\mathbf{where}{}\<[E]%
\\
\>[3]{}\hsindent{2}{}\<[5]%
\>[5]{}\Varid{return}{}\<[13]%
\>[13]{}\mathbin{::}\Varid{a}\to \Varid{f}\;\Varid{a}{}\<[E]%
\\
\>[3]{}\hsindent{2}{}\<[5]%
\>[5]{}(\bind ){}\<[13]%
\>[13]{}\mathbin{::}\Varid{f}\;\Varid{a}\to (\Varid{a}\to \Varid{f}\;\Varid{b})\to \Varid{f}\;\Varid{b}{}\<[E]%
\ColumnHook
\end{hscode}\resethooks
A monad is a functor (type constructor which is compatible with \ensuremath{\mathbin{-}\to \mathbin{-}}) which is endowed with two basic primitives. Reading a term of type
\ensuremath{\Varid{f}\;\Varid{a}} as a computation returning a value of type \ensuremath{\Varid{a}}, then the
primitive \ensuremath{\Varid{return}} embeds pure values as computations, while \ensuremath{(\bind )}
(pronounced \emph{bind}) accounts for the sequencing of computations.

The next interface we will model with monoidal categories is that of
applicative functors.
\begin{hscode}\SaveRestoreHook
\column{B}{@{}>{\hspre}l<{\hspost}@{}}%
\column{3}{@{}>{\hspre}l<{\hspost}@{}}%
\column{5}{@{}>{\hspre}l<{\hspost}@{}}%
\column{12}{@{}>{\hspre}l<{\hspost}@{}}%
\column{E}{@{}>{\hspre}l<{\hspost}@{}}%
\>[3]{}\mathbf{class}\;\Conid{Functor}\;\Varid{f}\Rightarrow \Conid{Applicative}\;\Varid{f}\;\mathbf{where}{}\<[E]%
\\
\>[3]{}\hsindent{2}{}\<[5]%
\>[5]{}\Varid{pure}{}\<[12]%
\>[12]{}\mathbin{::}\Varid{a}\to \Varid{f}\;\Varid{a}{}\<[E]%
\\
\>[3]{}\hsindent{2}{}\<[5]%
\>[5]{}(\circledast){}\<[12]%
\>[12]{}\mathbin{::}\Varid{f}\;(\Varid{a}\to \Varid{b})\to \Varid{f}\;\Varid{a}\to \Varid{f}\;\Varid{b}{}\<[E]%
\ColumnHook
\end{hscode}\resethooks
If we flip arguments for the combinator \ensuremath{(\circledast)}, we see that
applicative functors are similar to monads, with the difference that
the ``following'' computation on the combinator \ensuremath{(\circledast)} cannot depend
dynamically on its input parameter \ensuremath{\Varid{a}}.

The third interface we cover is arrows.
\begin{hscode}\SaveRestoreHook
\column{B}{@{}>{\hspre}l<{\hspost}@{}}%
\column{3}{@{}>{\hspre}l<{\hspost}@{}}%
\column{5}{@{}>{\hspre}l<{\hspost}@{}}%
\column{12}{@{}>{\hspre}l<{\hspost}@{}}%
\column{E}{@{}>{\hspre}l<{\hspost}@{}}%
\>[3]{}\mathbf{class}\;\Conid{Arrow}\;(\ensuremath{\leadsto})\;\mathbf{where}{}\<[E]%
\\
\>[3]{}\hsindent{2}{}\<[5]%
\>[5]{}\Varid{arr}{}\<[12]%
\>[12]{}\mathbin{::}(\Varid{a}\to \Varid{b})\to (\Varid{a}\ensuremath{\leadsto}\Varid{b}){}\<[E]%
\\
\>[3]{}\hsindent{2}{}\<[5]%
\>[5]{}(\ensuremath{\ggg}){}\<[12]%
\>[12]{}\mathbin{::}(\Varid{a}\ensuremath{\leadsto}\Varid{b})\to (\Varid{b}\ensuremath{\leadsto}\Varid{c})\to (\Varid{a}\ensuremath{\leadsto}\Varid{c}){}\<[E]%
\\
\>[3]{}\hsindent{2}{}\<[5]%
\>[5]{}\Varid{first}{}\<[12]%
\>[12]{}\mathbin{::}(\Varid{a}\ensuremath{\leadsto}\Varid{b})\to ((\Varid{a},\Varid{c})\ensuremath{\leadsto}(\Varid{b},\Varid{c})){}\<[E]%
\ColumnHook
\end{hscode}\resethooks
This interface builds on profunctors rather than on functors.  Terms
of type \ensuremath{\Varid{a}\ensuremath{\leadsto}\Varid{b}} are read as computations that ask for input of type
\ensuremath{\Varid{a}} and return output of type \ensuremath{\Varid{b}}. Using the \ensuremath{\Varid{arr}} combinator, one can
construct a computation from a function. The combinator \ensuremath{(\ensuremath{\ggg})}
sequences two computations. Finally, the combinator \ensuremath{\Varid{first}} represents
the \emph{strength} of the arrow: it is used to construct computations that
operate on an extra parameter without modifying it. The \ensuremath{\Varid{first}}
operation is handling the compatibility of the arrow with the
cartesian structure.

\begin{remark}\label{rk:strengh-implementation}
Morally, one can think that there is nothing corresponding to the
\ensuremath{\Varid{first}} operation in the type-classes \ensuremath{\Conid{Monad}} and \ensuremath{\Conid{Applicative}}
because a strength can be obtained for free in the case of functors:
\begin{hscode}\SaveRestoreHook
\column{B}{@{}>{\hspre}l<{\hspost}@{}}%
\column{E}{@{}>{\hspre}l<{\hspost}@{}}%
\>[B]{}\sigma \mathbin{::}\Varid{f}\;\Varid{a}\to \Varid{b}\to \Varid{f}\;(\Varid{a},\Varid{b}){}\<[E]%
\\
\>[B]{}\sigma \;\Varid{v}\;\Varid{b}\mathrel{=}\Varid{fmap}\;(\lambda \Varid{a}\to (\Varid{a},\Varid{b}))\;\Varid{v}{}\<[E]%
\ColumnHook
\end{hscode}\resethooks
We return later to this point in our discussion of strong profunctors.
\end{remark}

\subsection{Monoidal categories}
We recall briefly the basic definitions of monoidal categories. A
monoidal category is a category $\mathbb{C}$ together with a bifunctor
$\otimes : \mathbb{C} \times \mathbb{C} \to \mathbb{C}$, a
distinguished object $I$, and natural isomorphisms
\[
\lambda : I\otimes A \cong A,\quad
\rho : A\otimes I \cong A,\quad
\alpha : A\otimes (B\otimes C) \cong (A\otimes B)\otimes C
\]
such that certain coherences hold. The bifunctor $\otimes$ is referred
to as the \emph{tensor}, and the distinguished object $I$ as the
\emph{unit}. Moreover, we refer to the monoidal category by the triple
$(\mathbb{C}, \otimes, I)$, omitting $\lambda$, $\rho$ and
$\alpha$. When it is clear from the context, we even omit the monoidal
structure and just speak of $\mathbb{C}$ being a monoidal category.

Monoidal categories generalise categories with finite products: the
bifunctor is given by the binary product $\times$, the unit by the
terminal object, and the isomorphisms are defined as
\[
\begin{array}{lclcl}
  \lambda = \pi_2 & & \rho = \pi_1 & & \alpha = \langle \langle \pi_1, \pi_1 \circ \pi_2 \rangle, \pi_2 \circ \pi_2 \rangle \\
  \lambda^{-1} = \langle !_A, \id \rangle & & \rho^{-1} = \langle \id, !_A \rangle & & \alpha^{-1} = \langle \pi_1 \circ \pi_1, \langle \pi_2 \circ \pi_1, \pi_2 \rangle \rangle
\end{array}
\]%
where $\langle f , g \rangle : A \to B \times C$ is defined by the
universal property of products as the unique morphism such that $\pi_1
\circ \langle f, g \rangle = f$ and $\pi_2 \circ \langle f, g \rangle
= g$.

Another popular example of monoidal category is the category of
endofunctors $\left[\mathbb{C},\mathbb{C}\right]$ on any category
$\mathbb{C}$. The tensor is defined as composition of functors, and
the identity functor $\Id$ is the unit. This a \emph{strict} monoidal
category, i.e. a monoidal category in which the isomorphisms
$\lambda$, $\rho$ and $\alpha$ are identities.

Following the \emph{microcosm principle}, the notion of monoid can be
defined in any monoidal category $(\mathbb{C},\otimes,I)$. A monoid is
a triple $(M, m, e)$ where $M$ is an object of $\mathbb{C}$, and $m :
M \otimes M \to M$, $e : I \to M$ are morphisms of $\mathbb{C}$ such
that the equations
\[
\lambda = m \circ (e \otimes \id),\quad
\rho = m \circ (\id \otimes e),\quad
m \circ (m \otimes \id) \circ \alpha = m \circ (\id \otimes m)
\]
hold. For example, a monoid in $(\Set, \times, 1)$ is a monoid in the
usual mathematical sense, and a monoid in
$\left[\mathbb{C},\mathbb{C}\right]$ is a monad on
$\mathbb{C}$. Monoids on a monoidal category $(\mathbb{C}, \otimes,
I)$ form a category, where an morphism from a monoid $(M, m_M, e_M)$
to a monoid $(N, m_N, e_N)$ is a morphism $f : M \to N$ such that
\[
f \circ e_M = e_N\qquad\mbox{and}\qquad f \circ m_M = m_N \circ (f \otimes f).
\]
We denote this category by $\Mon{\mathbb{C}}$.

Let $(\mathbb{C},\otimes,I)$ and $(\mathbb{D},\oplus,J)$ be monoidal
categories. It is natural to consider functors that are compatible
with the monoidal structures. A \emph{strong monoidal functor} is a
triple $(F, \gamma, \gamma_0)$ where $F : \mathbb{C} \to \mathbb{D}$
is a functor, $\gamma_0 : J \cong F I$ is an isomorphism and
$\gamma_{A,B} : F A \oplus F B \cong F (A \otimes B)$ is family of
isomorphisms natural in $A$ and $B$ such that certain coherences
hold. If instead of isomorphisms we just consider morphisms, then we
have \emph{monoidal functors} if we ask the morphisms to go from left
to right, or \emph{opmonoidal functors} if the morphisms go from right
to left. A monoidal or opmonoidal functor can be \emph{normal} if the
morphism relating the units ($I$ and $J$) is invertible. We usually
omit the monoidal structure from $(F, \gamma, \gamma_0)$ and speak of
$F$ being a (strong, op)monoidal functor. However, it is important to
have in mind that a functor could have different monoidal structures.
 
\subsection{Profunctors and their strengths}\label{ss:profunctors-and-strength}
An endoprofunctor on the category of finite sets is a functor
$\Fin^\op \times \Fin \to \Set$.  By fixing arguments, each profunctor
on $\Fin$ comes with two canonical strengths. If we fix the
contravariant argument, we obtain a functor with a strength $\sigma^P$
in the usual sense (see \Cref{rk:strengh-implementation}), while in
the case of fixing the covariant argument, we obtain a strength in the
sense of Brady and Trimble~\cite{BT:CIPPL}. That is, we obtain a
contravariant functor $F : \Fin \to \Set$ with a family of maps
\[
\varsigma_{A,B} : F(A\times B)\times A \to F B,\qquad (v,a) \mapsto
F(\lambda z.(a, z))(v)
\]%
natural in $A$ and $B$, such that certain coherence conditions hold.
We denote these two strengths for an endoprofunctor $P : \Fin^\op \times \Fin \to \Set$ by
\[
\sigma^P : P(X,Y) \times Z \to P(X\times Z, Y),\qquad\qquad \varsigma^P : P(X\times Y,Z)\times X \to P(Y,Z).
\]
\begin{remark}
Note that such strengths are always available for functors $\Fin \to
\Set$. When we consider these strengths on the component $Z = 1$, the
strength morphism is actually invertible.
\end{remark}

The strengths $\sigma^P$ and $\varsigma^P$ are strengths independent
in each variable. Alternatively, we might consider strength for
endoprofunctors which act on both variables at the same time. This
gives the notion of \emph{strong endoprofunctor}, which is a form of
the notion of \emph{Tambara module} defined by Pastro and
Street~\cite{PS:DMC}.
\begin{definition}
  A endoprofunctor $P : \Fin^\op \times \Fin \to \Set$ is said to be \emph{(right)
  strong} if it comes equipped with a family of morphisms
  \[
  \str_{X,Y,Z} : P(X,Y) \to P(X\times Z, Y\times Z)
  \]
  natural in $X$, $Y$ and dinatural in $Z$ such that the equations
  \begin{gather*}
  P(\id, \pi_1) \circ \str_{X,Y,1} = P(\pi_1, \id), \\
  \str_{X,Y,W} \circ \str_{X,Y,V} = P(\alpha^{-1}, \alpha) \circ \str_{X,Y,V\times W}
  \end{gather*}
  hold.
  Since we work with the cartesian monoidal structure, there is no need
  to define left strong and bistrong endoprofunctors.
\end{definition}

Accordingly, a notion of morphism between strong endoprofunctors is
defined.

\begin{definition}
  A \emph{strong natural transformation} from $(P, \str^P)$ to $(Q,
  \str^Q)$ is a natural transformation $\tau : P \to Q$ such that the
  equation
  \[
  \tau_{X\times Z, Y\times Z} \circ \str^P_{X,Y,Z} = \str^Q_{X,Y,Z} \circ \tau_{X,Y}
  \]
  holds.
\end{definition}

Strong endoprofunctors and strong natural transformations between them
form a category, which we denote by $\ProfS$. There is a forgetful
functor
\[
U : \ProfS \to \Prof
\]%
which forgets the strength on a strong endoprofunctor. The functor $U$
has left and right adjoints, providing a way to construct free and
cofree strong endoprofunctors from an
endoprofunctor~\cite{RJ:NCM,PS:DMC}.

\subsection{Monoidal Structures on Finitary Functors}
We have already mentioned that the category of endofunctors forms a
monoidal category with composition, and that its monoids are
monads. In particular, a monoid in the monoidal category
$(\left[\Set,\Set\right], \circ, \Id)$ is a monad on $\Set$. A
finitary functor $F : \Fin \to \Set$ might be seen as a functor $F : \Set
\to \Set$ by the following construction:
\[
\bar{F}(Z) = \int^Y F(Y) \times (i~Y \to Z)
\]%
where $\inc : \Fin \to \Set$ denotes the inclusion functor from finite
sets to sets.

When we consider a finitary functors, i.e. the category $\Func$, we
obtain an alternative presentation of the composition structure, which
is defined by the coend formula
\[
 (F\circ G)(X) = \int^{C} F~C \times (C \to G X)
\]
and the unit is given by the inclusion functor $\inc$. A monoid for
this monoidal structure is a finitary monad on $\Set$, i.e. a finitary
functor $F$ together with morphisms
\[
e : \inc \to F,\qquad m : F \circ F \to F
\]
such that certain diagrams commute. The morphism $e$ corresponds to the
\ensuremath{\Varid{return}} operation in the \ensuremath{\Conid{Monad}} type-class, while $m$ corresponds to
the operation \ensuremath{(\bind )}.

There is another monoidal structure on finitary functors. As finitary
functors are presheaves, they have a canonical monoidal structure given
by the Day convolution. The tensor of finitary functors $F$ and $G$ as
objects in $\Func$ might be expressed as
\[
  (F \star G)(X) = \int^{C} F~C \times G~(C \to X),
\]%
and the inclusion functor $\inc : \Fin \to \Set$ acts again as the
unit. A monoid for this monoidal structure is given by a finitary
functor $F$ together with morphisms
\[
e : \inc \to F,\qquad m : F \star F \to F
\]
such that certain diagrams commute. The morphism $e$ represents the
\ensuremath{\Varid{pure}} operation of an idiom, and $m$ represents the combinator
\ensuremath{(\circledast)}.
\begin{remark}
  The categorical presentation of idioms as \emph{monoidal functors
    with a strength} is well known~\cite{MBP:APE}. We have chosen a
  different (but equivalent) presentation for the Day convolution as
  it is closer to the traditional \ensuremath{\Conid{Applicative}} interface. As the
  monoidal structure considered is cartesian, both constructions are
  equivalent.
\end{remark}

In contexts where confusion might arise, to distinguish between
$(\Func, \circ, \inc)$ and $(\Func, \star, \inc)$ without having to
mention all the monoidal structure, we simply refer to them as
$\Func_\circ$ and $\Func_\star$.

\subsection{Monoidal Structure on Strong Profunctors}
Profunctors in general, and endoprofunctors in particular, can be
composed by B\'enabou's tensor, giving rise to a monoidal structure on
the category of endoprofunctors $\Prof$. Indeed, this monoidal
structure can be lifted to strong endoprofunctors. We make this
structure precise in the following definition.

\begin{definition}
The category $\ProfS$ has a monoidal structure, with tensor product of
$(P, \str^P)$ and $(Q, \str^Q)$ defined as $(P\otimes Q,
\str^{(P\otimes Q)})$ where
\[
  (P \otimes Q)(X,Y) = \int^{W} P(X, W)\times Q(W, Y),
\]%
\[
  \str^{(P \otimes Q)}_{X,Y,Z} : \int^{W} P(X, W)\times Q(W, Y) \longrightarrow \int^{W} P(X \times Z, W) \times Q(W, Y\times Z).
\]%
The strength $\str^{(P \otimes Q)}_{X,Y,Z}$ is defined by the
universal property of coends as the unique morphism such that
\[
\str^{(P \otimes Q)}_{X,Y,Z} \circ \iota_W = \iota_{W\times Z} \circ (\str^P_Z \times \str^Q_Z).
\]%
with $\iota_W$ being the coend injection for $W$. The unit for the
monoidal tensor is $(\Hom, \str^\Hom)$, where
\[
\begin{array}[t]{lrcl}
\str^\Hom_{X,Y,Z} : & \Hom(X,Y) & \longrightarrow & \Hom(X\times Z, Y\times Z) \\
        & f & \longmapsto & f \times \id
\end{array}
\]%
\end{definition}

As observed by Jacobs et al.~\cite{JHH:CSA}, arrows correspond to
monoids in this category (disregarding the fact they consider a
general category $\mathbb{C}$ instead of $\mathbb{F}$). A monoid in
this category is a strong profunctor $(A, \str^A)$ together with
morphisms
\[
e : \Hom \to A,\qquad m : A \otimes A \to A
\]
such that certain diagrams commute. The morphism $e$ represents the
\ensuremath{\Varid{arr}} operation of the arrow, $m$ represents sequential composition
\ensuremath{\ensuremath{\ggg}}, and the strength $\str^A$ of the profunctor gives the \ensuremath{\Varid{first}}
operation.

In what follows, we sometimes refer to a strong profunctor $(P,
\str^P)$ just by $P$, but the reader has to have in mind that a
profunctor could have more than one strength, and that we expect
natural transformations to be strong.

\section{The Cayley and Kleisli Adjunctions}\label{sec:adjunctions}
The \Cref{eq:isoidiomintro,eq:isomonadintro} relate the categories of
monoids corresponding to
\[
\Func_\star,\qquad \ProfS,\qquad \Func_\circ.
\]
To obtain this result, we first construct adjunctions between the
underlying categories, and then see how they interact with the
monoidal structures. The central objects of our adjunctions are strong
profunctors, and how they can be recast into a functor. As we
explained in the introduction, the contravariant argument of the
profunctor might be fixed to the terminal object $1$ of $\Fin$,
obtaining thus a finitary functor. We capture the re-casting with the
functor $\ustar{\blank}$.
\[
  \begin{array}[t]{lcl}
    \ustar{\blank} & : & \ProfS \longrightarrow \Func \\
    \ustar{P} & = & Z \mapsto P(1, Z) \\
    \ustar{\tau}_Z & = & \tau_{1,Z}
  \end{array}
\]

The next step is to find left and right adjoints for the re-casting
functor. These adjoints might be derived using the left and right Kan
extensions~\cite{MacL:CWM}, although we do a direct presentation in
order to minimise the background material.

We begin describing the left adjoint to $\ustar{\blank}$. This functor
was called \emph{Cayley} in previous
work~\cite{PS:DMC,RJ:NCM,PSR:HNMC}. We now refer to it as
$\ladm{\blank}$, but keep the name Cayley to refer to the adjunction.

\begin{theorem}
  The functor $\ustar{\blank}$ has a left adjoint $\ladm{\blank}$ given by
  \[
  \begin{array}[t]{lcl}
    \ladm{\blank} & : & \Func \longrightarrow \ProfS \\
    \ladm{F} & = & ((X,Y) \mapsto F(X\to Y), v \mapsto F(\curry{(\ev \times \id) \circ \alpha})(v) ) \\
    \left(\ladm{\tau}\right)_{X,Y} & = & \tau_{X\to Y}
  \end{array}
  \]
  Moreover, the unit of the adjunction is a natural isomorphism.
\end{theorem}
\begin{proof}
  The unit and counit of the adjunction are
  \[
  \begin{array}[t]{lcl}
    \eta^{!}_F & : & F~Z \longrightarrow \ustar{\ladm{F}}~Z \\ 
    \eta^{!}_F & = & F(\curry{\rho})
  \end{array}
  \qquad\qquad
  \begin{array}[t]{lcl}
    \varepsilon^{!}_{(P, \str)} & : & \ladm{\ustar{P}}(X,Y) \longrightarrow P(X,Y) \\ 
    \varepsilon^{!}_{(P, \str)} & = & P(\lambda^{-1}, \ev) \circ \str
  \end{array}
  \]%
  where these definitions are parametric over $Z$ and $(X,Y)$
  respectively. After some calculations, it can be shown that the
  triangular identities hold.  The inverse of the unit is simply
  defined as $\eta^{!{-1}}_F = F(\ev \circ \rho^{-1})$.
\end{proof}

The right adjoint to $\lstar{\blank}$ is defined by a construction
similar to the direct image of a functor. In previous work we have
called this functor \emph{Kleisli}, and therefore we now refer to the
resulting adjunction as the \emph{Kleisli adjunction}.

\begin{theorem}
  The functor $\ustar{\blank}$ has a right adjoint $\lstar{\blank}$ given by
  \[
  \begin{array}[t]{lcl}
    \lstar{\blank} & : & \Func \longrightarrow \ProfS \\
    \lstar{F} & = & ((X,Y) \mapsto i~X \to F~Y, v \mapsto F(\curry{\sigma \circ (\ev \times \id) \circ \alpha})(v)) \\
    \left(\lstar{\tau}\right)_{X,Y} & = &  \curry{\tau_{Y} \circ \ev}
  \end{array}
  \]
  Moreover, the counit of the adjunction is a natural isomorphism.
\end{theorem}
\begin{proof}
    The unit and counit of the adjunction are
  \[
  \begin{array}[t]{lcl}
    \eta^{*}_{(P, \str)} & : & P(X,Y) \longrightarrow \lstar{\ustar{P}}(X,Y) \\ 
    \eta^{*}_{(P, \str)} & = & \curry{\varsigma^P \circ P(\rho, \id)}
  \end{array}
  \qquad\qquad
  \begin{array}[t]{lcl}
    \varepsilon^{*}_{F} & : & \ustar{\lstar{F}}~Z \longrightarrow F~Z \\ 
    \varepsilon^{*}_{F} & = & \ev \circ \rho^{-1}
  \end{array}
  \]%
  where these definitions are parametric over $(X,Y)$ and $Z$
  respectively.

  Again, it is routine to show that the triangular identities
  hold. The inverse of the counit is $\varepsilon^{*{-1}}_{F} = \curry{\rho}$.
\end{proof}

Summing up, we obtain the following situation, which we depict in a diagram.
\[
  \xymatrix@C+=3cm{
    \Func  \ar@/^{5mm}/[r]^{\ladm{\blank}}_{\bot} \ar@/_{5mm}/[r]_{\lstar{\blank}}^{\bot}
    & \ProfS \ar[l]|{\;\ustar{\blank}}
  }
\]

The fact that the unit (counit) is invertible can be expressed in a
number of equivalent statements. It will prove useful to ave these
equivalences around, and so we state the following
theorem~\cite{GZ:CFHT}.
\begin{theorem}\label{nlab-equivalences}
Let $L \vdash R : \mathbb{C} \to \mathbb{D}$ be an adjunction.

The following statements are equivalent:
\begin{itemize}
\item The counit $\varepsilon : LR \to \Id$ is a natural isomorphism.
\item $R$ is fully faithful.
\item The monad associated with the adjunction is idempotent.
\end{itemize}

Dually, the following are also equivalent:
\begin{itemize}
\item The unit $\eta : \Id \to RL$ is a natural isomorphism.
\item $L$ is fully faithful.
\item The comonad associated with the adjunction is idempotent.
\end{itemize}
\end{theorem}

We can apply Theorem~\ref{nlab-equivalences} to the adjunctions
$\ustar{\blank} \dashv \lstar{\blank}$ and $\ladm{\blank} \dashv
\ustar{\blank}$, and obtain the following propositions.
\begin{proposition}
The functor $\ladm{\blank}$ is fully faithful and the comonad $\Box
P = \ladm{\left(\ustar{P}\right)}$ is idempotent.
\end{proposition}
\begin{proposition}
The functor $\lstar{\blank}$ is fully faithful and the monad $\Diamond P
= \lstar{\left(\ustar{P}\right)}$ is idempotent.
\end{proposition}

\section{Equivalences from Monoidal Adjunctions}\label{sec:equivalences}
We now want to relate these adjunctions to the monoidal structures
presented in Section~\ref{sec:background}. The first step is to define
natural transformations and adjunctions that are compatible with
monoidal structures~\cite{AM:MFSHA}.
\begin{definition}
  Let $(F, \gamma, \gamma_0), (G, \delta, \delta_0) : (\mathbb{C},
  \otimes, I) \to (\mathbb{D}, \oplus, J)$ be monoidal functors. A
  natural transformation $\tau : F \to G$ is \emph{monoidal} if
  \[
   \tau_I \circ \gamma_0 = \delta_0,\qquad \tau_{A\otimes B} \circ \gamma_{A,B} = \delta_{A,B} \circ (\tau_A \oplus \tau_B).
   \]%
   There is a straightforward dual definition of \emph{opmonoidal
     natural transformations} between opmonoidal functors.
\end{definition}

\begin{definition}
  An adjunction $L \dashv R : \mathbb{C} \to \mathbb{D}$ with unit
  $\eta$ and counit $\varepsilon$, between monoidal categories
  $(\mathbb{C}, \otimes, I)$ and $(\mathbb{D}, \oplus, J)$, is a
  \emph{monoidal adjunction} if $L$ and $R$ are monoidal functors and
  $\eta$ and $\varepsilon$ are monoidal natural
  transformations. Similarly, the adjunction is \emph{opmonoidal} if
  both $L$ and $R$ are opmonoidal and $\eta$ and $\varepsilon$ are
  opmonoidal natural transformations. There is an intermediate case,
  sometimes called \emph{colax-lax monoidal adjunction}, which happens
  when $L$ is opmonoidal with structure $(\phi, \phi_0)$, $R$ is
  monoidal with structure $(\gamma, \gamma_0)$ and the following
  identities hold:
  \[
    R \phi_0 \circ \eta_I = \gamma_0,\qquad R \gamma_{A,B} \circ \eta_{A \otimes B} = \phi_{L A,L B} \circ \left(\eta_A \otimes \eta_B\right).
  \]
\end{definition}

The following theorem provides a way to transport a monoidal structure
from a functor to its adjoint and obtain a colax-lax monoidal
adjunction. This is a particular case of a phenomenon known as a
\emph{doctrinal adjunction}~\cite{Kel:DA}.

\begin{proposition}\label{prop:doctrinal}
  Let $L \dashv R$ be an adjunction between monoidal categories $(C,
  \otimes, I)$ and $(D, \oplus, J)$.

  If $L$ is an opmonoidal functor, then $R$ has structure of monoidal
  functor and the adjunction is colax-lax. Moreover, if $L$ is a
  strong monoidal functor, then $\eta$ and $\varepsilon$ are monoidal
  natural transformations.

  Dually, if $R$ is a monoidal functor, then $L$ has structure of
  opmonoidal functor and the adjunction is colax-lax. Moreover, if $R$
  is a strong monoidal functor, then $\eta$ and $\varepsilon$ are
  opmonoidal natural transformations.
\end{proposition}

The following propositions state the monoidal structure of the
functors involved in the Kleisli and Cayley adjunctions.
\begin{proposition}
  Considering $\Func$ with Day convolution and $\ProfS$ with
  B\'enabou's tensor, then the functor $\ladm{\blank}$ is strong
  monoidal, $\ustar{\blank}$ is monoidal and Cayley's adjunction is
  a monoidal adjunction.
\end{proposition}
\begin{proof}
  The strong monoidality of Cayley's functor was proven by Pastro and
  Street~\cite{PS:DMC}. Applying Proposition~\ref{prop:doctrinal}, we
  obtain the later results.
\end{proof}

In the case of the adjunction $\lstar{\blank} \dashv \ustar{\blank}$,
we obtain a ``weaker'' result.
\begin{proposition}
  Considering $\Func$ with substitution product and $\ProfS$ with
  B\'enabou's tensor, then the functor $\lstar{\blank}$ is monoidal,
  $\ustar{\blank}$ is opmonoidal and Kleisli's adjunction is a
  colax-lax adjunction.
\end{proposition}
\begin{proof}
  The monoidality of Kleisli's functor is well-known, see for
  instance~\cite{RJ:NCM}. Applying Proposition~\ref{prop:doctrinal},
  we obtain that $\lstar{\blank}$ is monoidal.
\end{proof}

We refer to the monoidal structure of $\lstar{\blank}$ by
$(\xi,\xi_0)$. What we meant above by ``weaker'' is that we do not
obtain that $\lstar{\blank}$ is strong monoidal. As we will see later,
this prevents some symmetry with the Cayley adjunction. Instead of
trying to obtain this condition by a careful analysis involving
finiteness and imposing extra conditions~\cite{ACU:MNE}, we rather
just work with monoidality of $\lstar{\blank}$, and put some extra
effort in the next section when we show that \emph{monads are
  promiscuous}.

The diagram of adjunctions in the previous section can now be
unfolded. The functors in the upper part are opmonoidal and the
functors in the lower part are monoidal.
\[
  \xymatrix@C+=3cm{
    \left(\Func, \star\right) \ar@{}[r]|{\bot} \ar@/^{5mm}/[r]^{\ladm{\blank}} & \left(\ProfS, \otimes \right) \ar@{}[r]|{\bot} \ar@/^{5mm}/[r]^{\ustar{\blank}} \ar@/^{5mm}/[l]^{\ustar{\blank}}  & \left(\Func, \circ\right) \ar@/^{5mm}/[l]^{\lstar{\blank}}
  }
  \]

We now give a categorical characterisation of monoids which are also
algebras (coalgebras) for an idempotent monad (comonad). This
definition provides a basis for presenting the right hand side of
\Cref{eq:isoidiomintro,eq:isomonadintro}.

\begin{definition}
  Given an idempotent monad (comonad) $T : \mathbb{C} \to \mathbb{C}$ on
  a monoidal category $(\mathbb{C}, \otimes, I)$, a \emph{$T$-monoid} is
  a quadruple $(M, m, e, \alpha)$ where
  \begin{itemize}
  \item $(M, m : M \otimes M \to M, e : I \to M)$ is a monoid in $\mathbb{C}$,
  \item $(M, \alpha : T M \to M)$ is a $T$-algebra ($(M, \alpha : M \to
    T M)$ is a $T$-coalgebra).
  \end{itemize}
\end{definition}

\begin{remark}
  Notice that there is no explicit coherence condition between the
  monoid structure and the $T$-algebra structure. However, the
  idempotency of the monad provides coherence between the two
  structures. For example, if $T$ is a strong monad, then a $T$-monoid
  in our sense is automatically a $T$-monoid in the sense of Fiore and
  Saville~\cite{FS:LOAS}.
\end{remark}

We can form a category of $T$-monoids by considering morphisms between
$T$-monoids which are morphisms as monoids and morphisms as
$T$-algebras between the underlying objects. However, as the monad
(comonad) is idempotent, the requirement of being a $T$-algebra
morphism becomes trivial. In fact, the $T$-algebra structure of a
$T$-monoid is like a property: an object $X$ can have at most one
$T$-algebra structure, and in case it does, it is $\eta_X^{-1}$
($\varepsilon_X^{-1}$). Despite that, we define a category of
$T$-monoids, as it will be useful to have a concrete name for it.

\begin{definition}
  Given an idempotent monad (comonad) $T$ on $\mathbb{C}$, the
  \emph{category of $T$-monoids}, $\Mon{T}$, consists of
  \begin{itemize}
  \item {\it Objects:} $T$-monoids.
  \item {\it Morphisms:} $T$-monoid homomorphisms, i.e. morphisms in
    $\mathbb{C}$ that are both monoid morphisms and $T$-algebra
    (coalgebra) morphisms at the same time.
  \end{itemize}
\end{definition}

From the comment above, it is clear that $\Mon{T}$ is equivalent to
the full subcategory of $\Mon{\mathbb{C}}$ such that the underlying
object of the monoid has a $T$-algebra structure.

We are now ready to revisit the results of Lindley et al., and obtain
their categorical analogues. We do so by postulating the equalities in
\Cref{eq:isoidiomintro,eq:isomonadintro} as equivalences between
categories.

\subsection{Idioms are Oblivious}

As already described, Lindley et al.~\cite{LWY:IOAMMP} characterise
idioms as the arrows in which the isomorphism
\[
\xymatrix@C+=2cm{
A\leadsto B \ar@/^5mm/[r]^-{\ensuremath{\Varid{delay}}} \ar@{}[r]|-{\quad\cong} & 1\leadsto (A \to B) \ar@/^5mm/[l]^-{\ensuremath{\Varid{force}}}
}
\]%
holds. The right to left morphism is available in any arrow as
\begin{hscode}\SaveRestoreHook
\column{B}{@{}>{\hspre}l<{\hspost}@{}}%
\column{E}{@{}>{\hspre}l<{\hspost}@{}}%
\>[B]{}\Varid{force}\mathbin{::}(\mathrm{1}\ensuremath{\leadsto}(\Varid{a}\to \Varid{b}))\to (\Varid{a}\ensuremath{\leadsto}\Varid{b}){}\<[E]%
\\
\>[B]{}\Varid{force}\mathrel{=}\lambda \Varid{f}.\ \Varid{arr}\;(\lambda \Varid{x}.\ ((),\Varid{x}))\ensuremath{\ggg}\Varid{first}\;\Varid{f}\ensuremath{\ggg}\Varid{arr}\;(\lambda (\Varid{f},\Varid{a}).\ \Varid{f}\;\Varid{a}){}\<[E]%
\ColumnHook
\end{hscode}\resethooks
From the arrow laws, it can be checked that \ensuremath{\Varid{force}}, as defined by
Lindley et al., is exactly the morphism given by the counit
\[
\varepsilon^{!}_{(P, \str^P)} : \ladm{\ustar{P}}(X,Y) = P(1, X \to Y) \longrightarrow P(X,Y) 
\]%
of the comonad $\Box$. Therefore, the isomorphism described in
\Cref{eq:isoidiomintro} can be expressed as a coalgebra for the
idempotent comonad $\Box$, i.e. an inverse for $\varepsilon^{!}$. The
other part of the equation, the arrow, might be seen as a monoid in
$\ProfS$ as we described in \Cref{sec:background}, thus obtaining that
the equivalent presentation of the equation:
\[
\textrm{Monoid in $\Func_\star$} = \textrm{$\Box$-monoid in $\ProfS$}\\
\]
To show this equivalence, we must be able to move monoids from one
category to another. Doing this through a monoidal functor is covered
by the following classical theorem (see for instance~\cite{AM:MFSHA}):
\begin{theorem}\label{thm:monoidal-preserves-monoids}
  A monoidal functor $(F,\phi) : (\mathbb{C},\otimes,I) \to
  (\mathbb{D},\oplus,J)$ can be lifted to a functor
  \[
  \begin{array}[t]{lcl}
    F^\bullet & : & \Mon{\mathbb{C}} \longrightarrow \Mon{\mathbb{D}}  \\
            &   & (M, m, e) \mapsto (F M, F m \circ \phi, F e \circ \phi_0)
  \end{array}
\]
\end{theorem}

A categorical explanation of the equation is given by the following
theorem, for which we give a direct proof using the theorem above.
\begin{theorem}\label{thm:idioms}
  Let $L \dashv R : \left(\mathbb{C}, \otimes, I\right) \to
  \left(\mathbb{D}, \oplus, J\right)$ be an adjunction in which $\eta$
  is an isomorphism, and $L$ is strong monoidal. Then,
  $\Mon{\mathbb{C}}$ is equivalent to $\Mon{LR}$.
\end{theorem}
\begin{proof}
  We give a direct proof of the equivalence. Using
  Proposition~\ref{prop:doctrinal}, we obtain a monoidal structure on
  $R$, and moreover, that the adjunction $L \dashv R$ is monoidal. As
  $L$ and $R$ are monoidal functors, they lift monoids, and give the
  following functors which form the equivalence.
\[
  \begin{array}[t]{lcl}
    L^\bullet & : & \Mon{\mathbb{C}} \longrightarrow \Mon{LR}  \\
            &   & (M, m, e) \mapsto (L M, L m \circ \phi, L e \circ \phi_0, L\eta)
  \end{array}
\]
\[
  \begin{array}[t]{lcl}
    R^\bullet & : &  \Mon{LR} \longrightarrow \Mon{\mathbb{C}}  \\
            &   & (M, m, e, \alpha) \mapsto (R M, R m \circ \psi, R e \circ \psi_0)
  \end{array}
\]
In addition, we need to provide natural isomorphisms $\beta :
\Id_{\Mon{LR}} \cong L^\bullet \circ R^\bullet$ and $\gamma :
\Id_{\Mon{\mathbb{C}}} \cong R^\bullet \circ L^\bullet$. For $\gamma$
we simply propose $\eta$, while for $\beta$ we use the $LR$-coalgebra.
\[
  \begin{array}[t]{l}
    \gamma_{(M,m,e)} : (M,m,e) \longrightarrow R^\bullet L^\bullet (M,m,e)  \\
    \gamma_{(M,m,e)} = \eta
  \end{array}\qquad
  \begin{array}[t]{l}
    \beta_{(M,m,e,\alpha)} : (M,m,e,\alpha) \longrightarrow L^\bullet R^\bullet (M,m,e,\alpha) \\
    \beta_{(M,m,e,\alpha)} = \alpha
  \end{array}
\]
That $\beta$ is well-defined comes from $\alpha = \varepsilon^{-1}$,
$\varepsilon$ being a monoidal natural transformation and $(M,
\alpha)$ being a coalgebra for $LR$.
\end{proof}

When we apply this theorem to the case $\mathbb{C} = \Func$,
$\mathbb{D} = \ProfS$, $L = \ladm{\blank}$ and $R = \ustar{\blank}$,
we obtain the semantic account of \Cref{eq:isoidiomintro}:
\[
\Mon{\Func}~\mbox{and}~\Mon{\Box}~\mbox{are equivalent categories.}
\]

We close the discussion on idioms by making the following
observation. The monoidal structure given by the Day convolution on
$\Func$ might be recovered from the one in $\ProfS$. As the comonad
$\Box$ is monoidal (being the composition of two monoidal
functors), we could lift the monoidal structure from $\ProfS$ to the
category of $\Box$-coalgebras~\cite{Moe:MTC}, which, as we noted
before, is equivalent to the category of finitary functors.

\subsection{Monads are Promiscuous}

In the case of monads, Lindley et al.~\cite{LWY:IOAMMP} characterise
them as the arrows in which the isomorphism
\[
\xymatrix@C+=2cm{
A\leadsto B \ar@/^5mm/[r]^-{\ensuremath{\Varid{eval}}} \ar@{}[r]|-{\quad\cong} & A\to (1 \leadsto B) \ar@/^5mm/[l]^-{\ensuremath{\Varid{lave}}}
}
\]%
holds. Instead of postulating an inverse to a morphism from left to
right, they follow Hughes~\cite{Hug:GMA} and ask for a mapping
\begin{hscode}\SaveRestoreHook
\column{B}{@{}>{\hspre}l<{\hspost}@{}}%
\column{E}{@{}>{\hspre}l<{\hspost}@{}}%
\>[B]{}\Varid{app}\mathbin{::}(\Varid{a}\ensuremath{\leadsto}\Varid{b},\Varid{a})\ensuremath{\leadsto}\Varid{b}{}\<[E]%
\ColumnHook
\end{hscode}\resethooks
which has to satisfy a set of three laws
\begin{align*}
\ensuremath{\Varid{first}\;(\Varid{arr}\;(\lambda \Varid{x}.\ \Varid{arr}\;(\lambda \Varid{y}.\ (\Varid{x},\Varid{y}))))\ensuremath{\ggg}\Varid{app}} &= \ensuremath{\Varid{arr}\;\Varid{id}} \\
\ensuremath{\Varid{first}\;(\Varid{arr}\;(\Varid{g}\ensuremath{\ggg}))\ensuremath{\ggg}\Varid{app}} &= \ensuremath{\Varid{second}\;\Varid{g}\ensuremath{\ggg}\Varid{app}} \\
\ensuremath{\Varid{first}\;(\Varid{arr}\;(\ensuremath{\ggg}\Varid{h}))\ensuremath{\ggg}\Varid{app}} &= \ensuremath{\Varid{app}\ensuremath{\ggg}\Varid{h}}
\end{align*}
These are the laws originally proposed by Hughes.  Notice that such a
mapping \ensuremath{\Varid{app}} might not be expressed in our framework. Naively, it
would be an element of $P(P(A,B)\times A, B)$, which is not
well-defined as $P(A,B)$ is not necessarily a finite set.

It can be shown, using the arrow laws, that giving a mapping \ensuremath{\Varid{app}}
which satisfies the three axioms above is equivalent to giving an
inverse to the mapping
\begin{hscode}\SaveRestoreHook
\column{B}{@{}>{\hspre}l<{\hspost}@{}}%
\column{E}{@{}>{\hspre}l<{\hspost}@{}}%
\>[B]{}\Varid{eval}\mathbin{::}(\Varid{a}\ensuremath{\leadsto}\Varid{b})\to (\Varid{a}\to (\mathrm{1}\ensuremath{\leadsto}\Varid{b})){}\<[E]%
\\
\>[B]{}\Varid{eval}\mathrel{=}\lambda \Varid{c}.\ \lambda \Varid{a}.\ \Varid{arr}\;(\lambda ().\ \Varid{a})\ensuremath{\ggg}\Varid{c}{}\<[E]%
\ColumnHook
\end{hscode}\resethooks
This map is already present in the translation of Lindley et al.\ in
the proof of the equivalence, and corresponds to the morphism given by
the unit
\[
\eta^{*}_{(P, \str^P)} : P(X,Y) \longrightarrow \lstar{\ustar{P}}(X,Y) = X \to P(1, Y)
\]%
of the monad $\Diamond$. Therefore, the isomorphism described in
\Cref{eq:isomonadintro} can be expressed as a algebra for
the idempotent monad $\Diamond$, i.e. an inverse for $\eta^{*}$. The arrow
part of the equation might be seen as a monoid in $\ProfS$, thus
obtaining that the equivalent presentation:
\[
\textrm{Monoid in $\Func_\circ$} = \textrm{$\Diamond$-monoid in $\ProfS$}\\
\]

To prove the equivalence induced by \Cref{eq:isomonadintro}, we might
want to use a kind-of-dual theorem to
Theorem~\ref{thm:idioms}. However, notice that this time the left
adjoint functor is not strong monoidal, but just opmonoidal, which
means that we cannot easily map monoids to monoids. We need to ask for
additional conditions to ensure that the left adjoint lifts
monoids. This case is covered by the following lemma, which is due to
Porst and Street~\cite{PS:GSD}.
\begin{lemma}\label{lem:oplax-preserves}
Let $(L, \varphi, \varphi_0) : (\mathbb{C}, \otimes, I) \to (\mathbb{D}, \oplus,
J)$ be an opmonoidal functor and $(C, m, e)$ a monoid in
$\mathbb{C}$. If
\begin{itemize}
\item $\varphi_0 : L I \to J$,
\item $\varphi_{C,C} : L (C \otimes C) \to L C \oplus L C$,
\item $\varphi_{C\otimes C, C} : L ((C \otimes C) \otimes C)  \to L (C \otimes C) \oplus L C$
\end{itemize}
are invertible, then $(L C, L m \circ \left(\varphi_{C,C}\right)^{-1}, L
e \circ \left(\varphi_0\right)^{-1})$ is a monoid in $\mathbb{D}$.
\end{lemma}

Using this lemma, we can prove the following theorem, which underlies
the basic structure of \Cref{eq:isomonadintro}.
\begin{theorem}
  Let $L \dashv R : \left(\mathbb{C}, \otimes, I\right) \to
  \left(\mathbb{D}, \oplus, J\right)$ be an adjunction in which
  $\varepsilon$ is an isomorphism, $(R, \gamma, \gamma_0)$ is a monoidal functor
  and the morphisms
  \[
  L\gamma_0 : L I \to LR J,\qquad\qquad
  L\left(\gamma \circ (\eta \otimes \id)\right) : L(C \otimes R D) \to LR (L C \oplus D)
  \]%
  are invertible. Then, $\Mon{\mathbb{D}}$ is equivalent to $\Mon{RL}$.
\end{theorem}
\begin{proof}
  Again, we give a concrete description of the equivalence. As
  $R$ is monoidal, it preserves monoids, so we use
  Theorem~\ref{thm:monoidal-preserves-monoids} to lift it to monoids.
  \[
  \begin{array}[t]{lcl}
    R^\bullet & : & \Mon{\mathbb{D}} \longrightarrow \Mon{RL} \\ & &
    (M, m, e) \mapsto (R M, R m \circ \gamma, R e \circ \gamma_0, R\varepsilon)
  \end{array}
  \]
  The inverse function this time is a bit more complex, as $L$ is not
  monoidal, and therefore does not automatically preserve monoids. We
  use Lemma~\ref{lem:oplax-preserves} to transport monoids. For a
  $M$-monoid $(M, m, e, \alpha)$, we can construct the required inverses as
  \begin{itemize}
  \item $\inv{\varphi_0} = \inv{L\gamma_0} \circ \varepsilon^{-1}$,
  \item $\inv{\varphi_{C,C}} = L(\id \otimes \alpha) \circ \inv{L\left(\gamma \circ (\eta \otimes \id)\right)} \circ \varepsilon^{-1}$,
  \item $\inv{\varphi_{C\otimes C, C}}= L(\id \otimes \alpha) \circ \inv{L\left(\gamma \circ (\eta \otimes \id)\right)} \circ \varepsilon^{-1}$.
  \end{itemize}

  We can apply Lemma~\ref{lem:oplax-preserves}, and obtain a functor $L^\bullet$:
  \[
  \begin{array}[t]{lcl}
    L^\bullet & : & \Mon{RL} \longrightarrow \Mon{\mathbb{D}}  \\
            &   & (M, m, e, \alpha) \mapsto (L M, L m \circ \left(\varphi_{C,C}\right)^{-1}, L e \circ \left(\varphi_0\right)^{-1})
  \end{array}
  \]%
  Notice that each $\left(\varphi_{C,C}\right)^{-1}$ depends on an
  $\alpha$. To prove the equivalence, we need to provide natural
  isomorphisms $\gamma : \Id_{\Mon{\mathbb{D}}} \cong L^\bullet \circ
  R^\bullet$ and $\beta : R^\bullet \circ L^\bullet \cong
  \Id_{\Mon{RL}}$. We propose the following.
  \[
  \begin{array}[t]{l}
    \beta_{(M,m,e)} : L^\bullet R^\bullet (M,m,e) \longrightarrow (M,m,e) \\
    \beta_{(M,m,e)} = \varepsilon
  \end{array}\qquad
  \begin{array}[t]{l}
    \gamma_{(M,m,e,\alpha)} : R^\bullet L^\bullet (M,m,e,\alpha) \longrightarrow (M,m,e,\alpha) \\
    \gamma_{(M,m,e,\alpha)} = \alpha
  \end{array}
  \]%
  These are invertible, as from assumptions $\varepsilon$ is
  invertible and $\alpha$ has $\eta$ as its inverse.
\end{proof}

The following proposition provides the missing hypotheses to use the
theorem above in the case we are interested.
\begin{proposition}
  The morphisms
  \[
  \ustar{\xi_0} : \ustar{\Hom} \to \ustar{\left(\lstar{\inc}\right)},\qquad\qquad
  \ustar{\left(\xi \circ (\eta^{*} \otimes \id)\right)} : \ustar{\left(P \otimes \lstar{F}\right)} \to \ustar{\left(\lstar{\left(\ustar{P} \oplus F\right)}\right)}
  \]%
  are invertible.
\end{proposition}
\begin{proof}
That these morphisms are invertible can be better seen by unfolding
the functor definitions on a finite set $Z$:
\[
\ustar{\xi_0}_Z : \Hom(1, Z) \to (1 \to Z)
\]
\[
\ustar{\left(\xi \circ (\eta^{*} \otimes \id)\right)}_Z : \int^{W} P(1, W)\times (i W \to F Z) \to \left(1 \to \int^{W} P(1, W)\times (i W \to F Z) \right)
\]
As the morphisms are only using the component $1$, they are invertible
(see remark on Subsection~\ref{ss:profunctors-and-strength}).
\end{proof}

Finally, we obtain the categorical counterpart of
\Cref{eq:isomonadintro} by applying this theorem to $\mathbb{C} =
\ProfS$, $\mathbb{D} = \Func$, $L = \ustar{\blank}$, $R =
\lstar{\blank}$:
\[
\Mon{\Func}~\mbox{and}~\Mon{\Diamond}~\mbox{are equivalent categories.}
\]

Note that in this case we cannot make a similar observation to the one
at the end of the previous subsection. The monad $\Diamond$ is not
monoidal, and therefore, we cannot (easily) lift the monoidal
structure from profunctors to functors. This is, partly, the symmetry
that gets broken from $\lstar{\blank}$ not being strong monoidal.

\section{Conclusions and Further Work}\label{sec:conclusions}

The relationship between different interfaces for computational
effects has been studied from different perspectives. The purely
syntactic approach has been covered by Lindley et
al.~\cite{LWY:IOAMMP}. A more programmatic account is covered by
Haskell's libraries, where the different interfaces are connected by
deriving type-class instances. On the other hand, the semantic point
of view for the connection was not very developed. Generally, the
semantic interpretation for notions of computation focusses only on
one of the interfaces, disregarding the connections between them. In
this paper we have taken a first step towards the connection from the
perspective of \emph{notions of computation as monoids}~\cite{RJ:NCM}.

As already remarked in the introduction, a direction of future work is
to replace the category $\Fin$ by a general category $\mathbb{C}$ with
an injection $i : \mathbb{C} \to \Set$, subject to conditions to be
determined, and probably similar to those used in \emph{relative
  monads}~\cite{ACU:MNE}. A related problem is to understand these
constructions in terms of Freyd categories and their
variations~\cite{LPT:MECBVPL,Atk:WCMA}, as they give a semantics for
arrows that do not suffer from size issues.

\begin{ack}
  The author is thankful to Soichiro Fujii, Pierre-Louis Curien,
  Marcelo Fiore and Ignacio L\'opez Franco for discussions on related
  topics, as well as to Hans-E. Porst and Ross Street for quick
  clarification on their results and pointing to related
  bibliography. The author is also grateful to the anonymous reviewers
  for their comments and suggestions.
\end{ack}

\bibliographystyle{eptcs}
\bibliography{generic}

\end{document}